\documentclass[11pt,letter]{article}
\usepackage{graphicx}
\usepackage[normalem]{ulem}
\usepackage{amssymb}
\usepackage{amsfonts}
\usepackage{amsmath}
\usepackage[all]{xy}
\usepackage[noend]{algorithmic}
\usepackage{hyperref}
\usepackage{verbatim} 
\usepackage{color}

\def\dnsbibitem{\vspace{-5pt}\bibitem}

\def\EQ{\mbox{\tt EQ}}
\def\DOM{\mbox{\tt DOM\_Part}}
\def\LAB{\mbox{\tt MAX\_Label}}
\def\ID{\mbox{\tt ID}}

\def\cP{\mathcal{P}}

\def\cD{\mathcal{D}}
\def\cE{\mathcal{E}}
\def\cH{\mathcal{H}}
\def\cA{\mathcal{A}}
\def\cD{\mathcal{D}}
\def\cC{\mathcal{C}}
\def\cJ{\mathcal{J}}
\def\cF{\mathcal{F}}
\def\hC{\hat{\mathcal{C}}}
\def\tO{{\tilde{O}}}
\def\INPUT{\bar{X}^s,\bar{X}^r}
\def\INPUTPrime{\bar{X'}^s,\bar{X'}^r}
\def\INPUTS{\bar{X}^s}
\def\INPUTR{\bar{X}^r}
\def\indVarS{\bar{X}^s}
\def\indVarR{\bar{X}^r}
\def\partyA{P_A}
\def\partyB{P_B}
\def\sminus{\smallsetminus}

\newtheorem{theorem}{Theorem}[section]
\newtheorem{corollary}[theorem]{Corollary}
\newtheorem{lemma}[theorem]{Lemma}

\newtheorem{observation}[theorem]{Observation}

\newtheorem{claim}[theorem]{Claim}

\newcommand{\qed}{\rule{7pt}{7pt}}
\newenvironment{proof}{\noindent{\bf Proof}\hspace*{1em}}{\qed\bigskip}
\newenvironment{proof-sketch}{\noindent{\bf Sketch of Proof}
\hspace*{1em}}{\qed\bigskip}
\newenvironment{proof-idea}{\noindent{\bf Proof Idea}
\hspace*{1em}}{\qed\bigskip}
\newenvironment{proof-of-lemma}[1]{\noindent{\bf Proof of Lemma #1}
\hspace*{1em}}{\qed\bigskip}
\newenvironment{sketch-of-claim}[1]{\noindent{\bf Proof Sketch of Claim #1}
\hspace*{1em}}{\qed\bigskip}
\newenvironment{proof-of-claim}[1]{\noindent{\bf Proof of Claim #1}
\hspace*{1em}}{\qed\bigskip}
\newenvironment{proof-claim}{\noindent{\bf Proof of Claim}
\hspace*{1em}}{\qed\bigskip}
\newenvironment{proof-lemma}{\noindent{\bf Proof of Lemma}
\hspace*{1em}}{\qed\bigskip}
\newenvironment{sketch-of-lemma}[1]{\noindent{\bf Proof Sketch of Lemma #1}
\hspace*{1em}}{\qed\bigskip}

\title{{\bf\Large Tight Bounds For Distributed Minimum-weight Spanning Tree
Verification}}

\date{\today}

\author{
Liah Kor
\thanks{Department of Computer Science and Applied Mathematics, The Weizmann
Institute of Science, Rehovot, 76100 Israel. 
E-mail: {\tt \{liah.kor,david.peleg\}@weizmann.ac.il}.
Supported by a grant from the United States-Israel Binational Science 
Foundation (BSF).}
\and
Amos Korman
\thanks{CNRS and LIAFA, Univ. Paris 7, Paris, France. 
E-mail: {\tt amos.korman@liafa.univ-paris-diderot.fr}. 
Supported by the ANR project DISPLEXITY, 
and by the INRIA project GANG.}
\and 
David Peleg~$^*$
}

\date{}
\begin{document}

\maketitle

\begin{abstract}
This paper introduces the notion of distributed verification without 
preprocessing. It focuses on the Minimum-weight Spanning Tree (MST) 
verification problem and establishes tight upper and lower bounds for 
the time and message complexities of this problem. Specifically, we
provide an MST verification algorithm that achieves {\em simultaneously} 
$\tilde{O}(m)$ messages and $\tilde{O}(\sqrt{n} + D)$ time, where $m$ is 
the number of edges in the given graph $G$, $n$ is the number of nodes, 
and $D$ is $G$'s diameter. On the other hand, we show that any MST 
verification algorithm must send $\tilde{\Omega}(m)$ messages and incur 
$\tilde{\Omega}(\sqrt{n} + D)$ time in worst case.

Our upper bound result appears to indicate that the verification of an MST may be easier than 
its construction, since for MST construction, both lower bounds of
$\tilde{\Omega}(m)$ messages and $\tilde{\Omega}(\sqrt{n} + D)$ time hold, but at the moment 
there is no known distributed algorithm that constructs an MST and  achieves 
{\em simultaneously} $\tilde{O}(m)$ messages and $\tilde{O}(\sqrt{n} + D)$ time.
Specifically, the best known time-optimal algorithm (using $\tO(\sqrt{n} + D)$ 
time) requires $O(m+n^{3/2})$ messages, and the best known message-optimal 
algorithm (using $\tO(m)$ messages) requires $O(n)$ time.
On the other hand, our lower bound results indicate that the verification of an MST is not significantly easier than its construction.
\end{abstract}


\noindent {\bf keywords:} Distributed algorithms, distributed verification, labeling schemes, minimum-weight spanning tree.

\section {Introduction}

\subsection{Background and Motivation}
The problem of computing a Minimum-weight Spanning Tree (MST)  received considerable attention in both the 
distributed setting as well as in the  centralized setting. In the distributed setting, constructing such a tree distributively
requires a collaborative computational effort by all the network vertices, 
and involves sending messages to remote vertices and waiting for their replies. This is particularly interesting in the $\cal{CONGEST}$ model of computation, where due to congestion constrains, each message can encode only a limited number of bits, specifically, this number is typically assumed to be $O(\log n)$, where $n$ denotes the number of nodes in the network.
The main measures considered for evaluating a distributed MST protocol are the {\em message} complexity, namely, the maximum 
number of messages sent in the worst case scenario, 
and the {\em time} complexity, namely, the maximum number of communication rounds 
required for the protocol's execution in the worst case scenario.
The line of research on the distributed MST computation problem was initiated 
by the seminal work of Gallager, Humblet, and Spira \cite{GHS_83} 
and culminated in the $O(n)$ time and $O(m+n\log n)$ messages algorithm 
by Awerbuch \cite{A87}, where $m$ denotes the number of edges in the network. 
Both of these algorithms apply also to asynchronous systems.
As pointed out in \cite{A87}, the results of 
\cite{AGVP_90,B_80,FL_84} establish an $\Omega(m+n\log n)$ lower bound 
on the number of messages required to construct a MST.  
Thus, the algorithm of \cite{A87} is essentially optimal. 

This was the state of affairs until the mid-nineties when Garay, Kutten, 
and Peleg \cite{GKP98} initiated the analysis of the time complexity 
of MST construction as a function of additional parameters (other than $n$), 
and gave the first sublinear time distributed algorithm for the MST problem, 
running in time $O(D+n^{0.614})$, where $D$ is the diameter of the network. 
This result was later improved to $O(D + \sqrt{n} \log^* n)$ by 
Kutten and Peleg \cite{KP_98}.
The tightness of this latter bound was shown by  
Peleg and Rubinovich \cite{PR_00} who proved that 
$\tilde{\Omega}(\sqrt{n})$ is essentially\footnote{$\tilde{\Omega}$ 
(respectively $\tilde{O}$) is a relaxed variant of the $\Omega$ (rep., $O$) 
notation that ignores polylog factors.}
a lower bound on the time for constructing MST on graphs with
diameter $\Omega(\log n)$. This result was complemented by the work 
of Lotker,  Patt-Shamir and  Peleg \cite{LPP_01} that showed an 
$\tilde{\Omega}(\sqrt[3]{n})$ lower bound on the time required for 
MST construction on graphs with small diameter. 
Note, however, that the time-efficient algorithms of \cite{GKP98,KP_98} 
are not message-optimal, i.e., they take asymptotically much more than 
$O(m+n \log n)$ messages. For example, the time-optimal protocol 
of \cite{KP_98} requires sending $O(m + n^{3/2})$ messages. 
The question of whether there exists an optimal distributed algorithm for MST 
construction that achieves {\em simultaneously} $\tilde{O}(m)$ messages 
and $\tilde{O}(\sqrt{n} + D)$ time still remains open.

This paper addresses the MST verification problem in the $\cal{CONGEST}$ model. 
Informally, the setting we consider is as follows. 
A subgraph is given in a distributed manner, namely, some of 
the edges incident to every vertex are locally marked, and the collection of 
marked edges at all the vertices defines a {\em marked subgraph}; see, e.g., 
\cite{CGKK95,GHS_83,KK07,KKP10}. The verification task requires checking 
distributively whether the marked subgraph is indeed an MST of the given graph. 
Similarly to the papers dealing with sub-linear time MST construction 
\cite{GKP98,KP_98}, we consider a synchronous environment.

\subsection{Our Results}

We consider the MST verification problem and establish asymptotically tight 
upper and lower bounds for the time and message complexities of this problem.
Specifically, in the positive direction we show the following:

\begin{theorem}
\label{thm:upper-bound}
There exists a distributed MST verification algorithm that uses
$\tO(\sqrt{n} + D)$ time and $\tO(m)$ messages.
\end{theorem}

This result appears to indicate that MST verification may be easier than 
MST construction, since, at the moment,  it is not known whether there exists an algorithm that constructs an MST simultaneously in $\tO(\sqrt{n} + D)$ time and $\tO(m)$ messages. 
Conversely, we show that the verification problem 
is not much easier than the construction, by proving that the known lower bounds for MST construction
also hold for the verification problem. Specifically, we prove the following two matching lower bounds.

\begin{theorem} 
\label{thm:lower-bound-msgs}
Any distributed algorithm for 
MST verification 
requires ${\Omega}(m)$ messages.
\end{theorem}

\begin{theorem} 
\label{thm:lower-bound-time}
Any distributed algorithm for 
MST verification 
requires $\tilde{\Omega}(\sqrt{n} + D)$ time.
\end{theorem}

To the best of our knowledge, this paper is the first to investigate
distributed verification  without assuming any kind of preprocessing (see Section~\ref{preprocessing}). Following this paper, several other works on distributed verification have already been published. More specifically, a systematic study of distributed verification is established for various verification tasks \cite{DHKKNPPW}. Distributed verification in the $\cal{LOCAL}$ model has been studied in \cite{FKP11,FKPP12}, mostly focusing on computational complexity issues.

Our $\tilde{\Omega}(\sqrt{n} + D)$ time lower bound for verifying an MST is achieved by a (somewhat involved) 
modification of the corresponding lower bound for the computational task 
\cite{PR_00}. 
The idea behind our time lower bound was already found useful in several 
continuation studies. Specifically, a modification of our proof technique 
was used in \cite{DHKKNPPW} to yield time lower bounds for several other 
verification tasks, and for establishing a lower bound on the hardness 
of approximating an MST. Somewhat surprisingly, in \cite{NDP11}, 
the proof technique was extended even further to yield a result 
in the (seemingly unrelated) area of distributed random walks.

\subsection{Other related work}
\subsubsection{MST computation and verification in the centralized setting}  In the centralized setting, there is a large body of literature concerning the
problem of efficiently computing an MST of a given weighted graph.
Reviews can be found, e.g., in the survey paper by Graham and Hell \cite{GH85} 
or in the book by Tarjan \cite{T83} (Chapter 6). 
The fastest known algorithm for finding an MST is that of 
Pettie and Ramachandran  \cite{optimalMST}, which runs
in $O(m\cdot\alpha(m,n))$ time, where $\alpha$ is 
 the inverse Ackermann function, 
$n$ is the number of vertices 
and $m$ is the number of edges in the graph. 
Unfortunately, a linear (in the number of edges) time algorithm 
for computing an MST is known only in certain cases, or by using randomization
\cite{FredmanWillard,KargerKleinTarjan}. 

The separation between computation and verification, and specifically,
the question of whether verification is easier than computation, is a central 
issue of profound impact on the theory of computer science. 
In the context of MST, the  verification problem 
(introduced by Tarjan \cite{Tarjan79}) is the following:
given a weighted graph, together with a subgraph, 
it is required to decide whether this subgraph forms an MST of the graph. 
At the time it was published, the running time of the MST verification 
algorithm of \cite{Tarjan79} was indeed superior to the best known bound
on the computational problem. Improved verification  algorithms in different
centralized models were then given by Harel \cite{Harel85},
Koml$\grave{o}$s \cite{Komlos}, and Dixon, Rauch, and Tarjan
\cite{DixonRauchTarjan}, and parallel algorithms were presented by Dixon and
Tarjan \cite{DixonTarjan} and by King, Poon, Ramachandran, and Sinha
\cite{king2}. 
Even though it is not known whether there exists a deterministic algorithm that 
computes an MST in $O(m)$  time, the verification algorithm of 
\cite{DixonRauchTarjan} is in fact linear, i.e., runs in time $O(m)$ (the same result with a simpler 
algorithm was later presented by King \cite{king97} and by Buchsbaum \cite{BGKRTW_08}). 
For the centralized setting, this may indicate that the verification 
of an MST is indeed easier than its computation.

\subsubsection{Distributed verification with preprocessing}\label{preprocessing}

Some previous 
papers investigated distributed verification tasks assuming that the algorithm 
designer can perform a preprocessing stage to help the verification, 
cf. ~\cite{AKY97,APV,silent,KK07,KKP10}. 
Typically, in this preprocessing stage, data structures (i.e., labels, proofs) 
are provided to the nodes of the graph, and using these data structures, 
the verification proceeds in a constant number of rounds. The focus in those 
studies was on the  minimum size of a data structure (i.e., amount of 
information stored locally), while the complexities of the preprocessing stage 
providing the data structures were not analyzed. 

 \section{Preliminaries}

\subsection{The Model}

A point-to-point communication network is modeled as an undirected
graph $G(V,E)$,
where the vertices in $V$ represent the network processors and the
edges in $E$ represent the communication links connecting them. We denote by $n$ the number of vertices in $G$, i.e., $n=|V|$, and let $m$ denote the number of edges, i.e., $m=|E|$.
The {\em length} of a path in $G$  is the number of edges it
contains. The {\em distance} between two vertices $u$ and $v$ is
 the length of the shortest path connecting them. 
The {\em diameter} of $G$, denoted $D$, is the maximum distance
between any two vertices of $G$.

Vertices are assumed to have unique {\em identifiers}, 
and each vertex $v$ knows its own identifier $\ID(v)$.
A weight function $\omega: E \rightarrow \mathbb{N}$ associated with the
graph assigns a nonnegative integer {\em weight} $\omega(e)$ to each edge
$e=(u,v)\in E$.  The vertices do not know the topology or the edge weights of the
entire network, but they  know the weights of the edges incident to them, that is, the weight $\omega((u,v))$ is known to  $u$ and $v$. 
Similarly to corresponding works on MST computation,  we assume that the edge weights are bounded by a polynomial in $n$ 
(this assumption implies that a weight of an edge can be encoded using $O(\log n)$ bits, and hence can be encoded in one message).

Similarly to previous work, we consider the $\cal{CONGEST}$ model. 
Specifically, the vertices can communicate only by sending and receiving
messages over the communication links. 
Each vertex can distinguish between its incident edges. Moreover, 
if vertex $v$ sends a message to vertex $u$ along the edge $e=(v,u)$,
then upon receiving the message, vertex $u$ knows that the message 
was delivered over the edge $e$. At most one $O(\log n)$ bits can be sent 
on each link in each message. 
Similarly to  \cite{GKP98,KP_98}, we assume that the communication 
is carried out in a synchronous manner,
i.e., all the vertices are driven by a global clock.
Messages are sent at the beginning of each round, and are received at
the end of the round.
(Clearly, our lower bounds hold for asynchronous networks as well.)
Relevant studies typically assume that computations start 
either at a single {\em source} node or simultaneously at all nodes.
Our results hold for both of these settings.

\subsection{The distributed MST Verification problem}

Formally, the {\em minimum-weight spanning tree (MST) verification} problem can be 
stated as follows. Given a graph $G(V,E)$, a weight function $\omega$ 
on the edges, and a subset of edges $T\subseteq E$, referred as the 
{\em MST candidate},  it is required to decide whether  $T$ forms a minimum  
spanning tree on $G$, i.e., a spanning tree  whose total weight 
$\omega(T)= \sum_{e\in T} \omega(e)$ is minimal.
In the distributed model, the input and output of the 
{\em MST verification problem} are represented as follows.
Each vertex knows the weights of the edges connected to its 
immediate neighbors. A degree-$d$ vertex $v \in V$ with
neighbors $u_1,\dots, u_d$ has $d$ {\em weight variables}
$W_1^v, \dots, W_d^v$, with $W_i^v$ containing the weight of the edge
connecting $v$ to $u_i$, i.e., $W_i^v = \omega(v, u_i)$, and $d$ 
{\em boolean indicator variables}  $Y_1^v, \dots, Y_d^v$ indicating which 
of the edges adjacent to $v$  participate in the MST candidate that we wish 
to verify. The indicator variables must be consistent, namely, 
for every edge $(u,v)$, the indicator variables stored at $u$ and $v$ 
for this edge must agree (this is easy to verify locally). 
Let $T_Y$ be the set of edges marked by the indicator variables 
(i.e., all edges for which the indicator variable is set to~1).
The output of the algorithm at each vertex $v$ 
is an assignment to a (boolean) output variable $A^v$ that must satisfy $A^v=1$
if $T_Y$ is an MST of $G(V,E,\omega)$, and $A^v=0$ otherwise.

\section{An MST Verification Algorithm}
\label{sec:alg}

In this section we describe our MST verification algorithm, 
prove its correctness and analyze its time and message complexities.

First, we note that in this section we assume that the verification algorithm 
is initiated by a designated source node. The case in which all nodes wake up 
simultaneously can be reduced to the single source setting using the 
leader election algorithm of \cite{KPPRT12}, which employs $\tilde{O}(m)$ 
messages and runs in $\tilde{O}(D)$ time.

\subsection{Definitions and Notations}

Following are some definitions and notations used in the description 
of the algorithm.
For a graph $G=(V,E,\omega)$, an edge $e$ is said to be {\em cycle-heavy} 
if there exists a cycle $C$ in $G$ that contains $e$, and $e$ has 
the heaviest weight in $C$.
For a graph $G=(V,E,\omega)$, a set of edges $F\subseteq E$ is said to be an 
{\em MST fragment} of $G$ if there exists a minimum spanning tree $T$ 
of $G$ such that $F$ is a subtree of $T$ 
(i.e., $F\subseteq T$ and $F$ is a tree). 
Similarly, a collection  $\cF$ of edge sets is referred to as an 
{\em MST fragment collection} 
of $G$ if there exists an MST  $T$ of $G$ 
such that 
(1) $F_i$ is a subtree of $T$  for every $F_i\in \cF$,
(2) $\bigcup_{F_i\in \cF} V(F_i)=V$,
and 
(3) $V(F_i) \cap V(F_j) =\emptyset$ for every $F_i,F_j\in \cF$, where $i\neq j$.

Consider a graph $G=(V,E,\omega)$, an MST fragment collection $\cF$,
a subgraph $T$ of $G$ and a vertex $v$ in $G$.
The {\em fragment graph} of $G$, denoted $G_{\cF}$, is defined as a graph 
whose vertices are the MST fragments $F_i\in \cF$, and whose edge set contains 
an edge $(F_i, F_j)$  if and only if there exist vertices $u\in V(F_i)$ and $v\in V(F_j)$ 
such that $(u,v)\in E$. Similarly,
the {\em fragment graph induced by} $T$, denoted $T_{\cF}$,  
is defined as a graph whose vertices are the MST fragments $F_i\in \cF$, 
and whose edge set contains an edge $(F_i, F_j)$ for $i\neq j$  if and only if there exist vertices 
$u\in V(F_i)$ and $v\in V(F_j)$ such that $(u,v)\in T$. The edges of $T_{\cF}$ 
are also referred to as the {\em inter-cluster} edges induced by $T$.

For each vertex $v$, let $E_T(v)$ denote the set of edges of $T$ incident to $v$.
For each fragment $F\in \cF$, the set of {\em fragment internal} 
edges of $F$ induced by $T$, denoted $E_{F,T}$, consists of all edges of $T$ with both 
endpoints in $V(F)$, i.e., 
$E_{F,T}=\{e\; |\; e=(u,v)\in T \; and \;  \ u,v\in V(F)\}$.
The fragment of $v$, denoted by $F(v)$, is the fragment 
$F\in \cF$ such that $v\in V(F)$.
Denote by $E_{F}(v)$ the set of edges in $F(v)$ that are incident to $v$ 
(i.e., $E_{F}(v)=\{e | \; e=(u,v)\in E\; and \;  u\in V(F(v))\}$).
Similarly, denote by $E_{F,T}(v)$ the set of fragment internal edges of $F$ induced by $T$ 
and incident to~$v$.

Throughout the description of the verification algorithm we assume that 
the edge weights are distinct. Having distinct edge weights simplifies our arguments   
since it guarantees the uniqueness of the MST. Yet,  this assumption is not essential. Alternatively, in case the graph is not guaranteed 
to have distinct edge weights, we may use a modified weight function 
$\omega'(e)$, which orders edge weights lexicographically. 
At this point we would like to note that the standard technique (e.g., \cite{GHS_83}) for obtaining unique weights is not sufficient for our purposes. 
Indeed, that technique orders edge weights lexicographically:  
first, by their original weight $\omega(e)$, and then, by the identifiers of 
the edge endpoints (say, first comparing the smaller of the two identifiers of the endpoints, and then the larger one). 
This yields a modified graph with unique edge weights, and an MST of the modified graph is necessarily an MST of the original graph. 
For construction purposes it is therefore sufficient to consider only 
the modified graph. However, this is not the case for verification purposes, 
as the given subgraph can be an MST of the original graph but not necessarily 
an MST of the modified graph. 

While we cannot guarantee than any MST of the original graph is an MST of 
the modified graph (having unique edge weights), we make sure that 
the particular given subgraph $T$ is an MST of the original graph 
if and only if it is an MST of modified one. This condition is sufficient for our purposes, and allows us to consider only the modified graph.
Specifically, to obtain the modified graph, we employ a slightly different 
technique than the classical one. 
That is, edge weights are lexicographically ordered as follows.
For an edge $e=(v,u_i)$ connecting $v$ to its $i$th neighbor $u_i$, 
consider first its original weight $\omega(e)$, 
next, the value $1-Y_i^v$ where $Y_i^v$ is the indicator variable of the edge
$e$ (indicating whether $e$ belongs to the candidate MST to be verified), 
and finally, the identifiers of the edge endpoints, $\ID(v)$ and $\ID(u_i)$
(say, first comparing the smaller of the two identifiers of the endpoints, 
and then the larger one). 
Formally, let 
$$\omega'(e) ~=~ 
\left\langle \omega(e), 1-Y_i^v, \ID_{min}(e), \ID_{max}(e) \right\rangle~,$$
where $\ID_{min}(e) = \min\{\ID(v),\ID(u_i)\}$
and $\ID_{max}(e) = \max\{\ID(v),\ID(u_i)\}$.
Under this weight function $\omega'(e)$, edges with indicator variable set to 1
will have lighter weight than edges with the same weight under $\omega(e)$ 
but with indicator variable set to 0 
(i.e., for edges $e_1\in T$ and $e_2\notin T$ such that 
$\omega(e_1)=\omega(e_2)$, we have $\omega'(e_1)< \omega'(e_2)$).
It follows that  the given subgraph $T$ 
 is an MST of $G(V,E,\omega)$ if and only if  $T$ is an MST
of $G(V,E,\omega')$. Moreover, since $\omega'(\cdot)$ takes into account  
the unique vertex identifiers, it assigns distinct edge weights.

The MST verification algorithm makes use of Procedures $\DOM$ and $\LAB$, 
presented in \cite{KP_98} and \cite{KKKP_05} respectively. Before we continue, let us first recall what these procedures guarantee. 

\paragraph{Procedure $\DOM$:} 
The distributed Procedure $\DOM$, used in \cite{KP_98}, partitions a given 
graph into an {\em MST fragment collection (MFC)} $\cF$, where each fragment 
is of size at least $k+1$ and depth $O(k)$, for a parameter $k$ to be 
specified later (aiming to optimize the total costs). 
A {\em fragment leader} vertex is associated with each constructed fragment 
(the identifier of the fragment is defined as the identifier of the 
fragment's leader). 
After Procedure $\DOM$ is completed, 
each vertex $v$ knows the identifier of the fragment 
to which it belongs  and $v$'s incident edges that belong to the fragment.
(To abide by the assumption of \cite{KP_98} that each vertex 
knows the identifiers of its neighbors, before applying Procedure $\DOM$, 
the algorithm performs a single communication round that exchanges 
vertex identifiers between neighboring vertices.) The execution of 
Procedure $\DOM$, requires $O(k\cdot \log ^*n)$ time and 
$O(m\cdot \log k +n\cdot \log^*n\cdot \log k)$ messages. 
The parameter $k$ to be used by our protocol, chosen as a suitable breakpoint 
so as to optimize the total costs, satisfies $k=\Theta(\sqrt n+D)$, 
hence, our invocation of Procedure $\DOM$ will require 
$\tilde{O}(\sqrt{n}+D)$ time and $\tilde{O}(m)$ messages.

\paragraph{The $\LAB$  labeling scheme:} 
The labeling scheme $\LAB$ 
of \cite{KKKP_05} is a centralized procedure designed for the family of 
weighted trees. The procedure involves two components: 
an encoder algorithm $\cE$ and 
a decoder algorithm $\cD$. These two components satisfy the following.
\begin{enumerate}
\item Given a weighted tree $T$, the encoder algorithm $\mathcal E$ 
assigns a label $L(v)$ to each vertex~$v$~of~$T$.
\item Given two labels $L(u)$ and $L(v)$ assigned by $\mathcal E$ to two vertices $u$ and $v$ in some weighted tree~$T$, the decoder 
algorithm $\mathcal D$ outputs $MAX(u,v)$, which is the maximum weight 
of an edge on the path connecting $u$ and $v$ in $T$. (The decoder 
algorithm $\mathcal D$ bases its answer on the pair of labels   $L(u)$ and $L(v)$ only, without knowing any additional information regarding the tree~$T$.)
\end{enumerate} 

The labeling scheme  $\LAB$  produces $O(\log n \log W)$-bit labels, 
where $W$ is the largest weight of an edge. Since $W$ is assumed to be 
polynomial in $n$, the label size is $O(\log ^2 n)$ bit.

\subsection{The algorithm}

\paragraph{Overview:} The algorithm consists of three  phases. The first phase  starts by letting the source node $s$ construct a BFS tree rooted at $s$, and calculating the values of the number of nodes $n$ and a 2-approximation $d$ to the diameter $D$ of the graph. Then, the distributed 
Procedure $\DOM$ of \cite{KP_98} in invoked with parameter $k=\max \{\sqrt n,d\}=\Theta(\sqrt n+D)$, where $D$ is the diameter of the graph. This procedure constructs an MST fragment collection 
(MFC) $\cF$, where every fragment in $\cF$
is of size at least $k+1$ and depth $O(k)$. As mentioned, our invocation of Procedure $\DOM$ requires $\tilde{O}(\sqrt{n}+D)$ time and 
$\tilde{O}(m)$ messages.
The algorithm verifies that the set of edges contained in the constructed MFC 
is equal to the set of fragment internal edges induced by the MST candidate $T$, i.e., 
$\bigcup_{F\in \cF} E(F) = \bigcup_{F\in \cF} E_{F,T}$ 
(this verifies that each $F\in \cF$ is contained in $T$ and that  $T$ does not contain additional fragment internal edges.)
Note that this is a necessary  condition for correctness since the graph is assumed to have  a unique MST. 

In the following phases, the algorithm verifies that all remaining edges of $T$  form an MST on the fragment graph $G_{\cF}$. Let $T_{\cF}$ be the fragment 
graph induced by $T$ with respect to the MFC $\cF$ found in the previous phase.
In order to verify that $T_{\cF}$ forms an MST on the 
fragment graph $G_{\cF}$, it suffices to verify that $T_{\cF}$ is a tree and 
that none of the edges of $T_{\cF}$ is  a {\em cycle heavy} edge in  $G_{\cF}$.
The above is done as follows. 

In the second phase, the structure of  $T_{\cF}$ is aggregated over the BFS tree to the source node $s$, which in turn verifies that
 $T_{\cF}$ is indeed a tree. Note that this aggregation is not very wasteful, and requires $O(D+f)$ time and $O(D\cdot f)$ messages, where $f$ is the number of nodes in $T_{\cF}$. As shown later, $f \le n/k$, and hence, this aggregation can be done using $\tilde{O}(\sqrt{n}+D)$ time and 
$\tilde{O}(m)$ messages.

In the third phase, the source node $s$ employs the (centralized) labeling scheme $\LAB$ of \cite{KKKP_05} (or \cite{KK07}) on $T_{\cF}$. 
Informally, the labeling scheme assigns a label $L(F_i)$ to each vertex $F_i$ of $T_{\cF}$ 
using  the encoder algorithm $\cE$ applied on $T_{\cF}$. The label $L(F_i)$ is then efficiently delivered to each vertex in $F_i$.
More specifically, the $f$ labels are broadcasted over the BFS tree, so that eventually, each node in a fragment knows its corresponding label. This broadcasts costs $\tilde{O}(D+f)=\tilde{O}(\sqrt{n}+D)$ time and $\tilde{O}(D\cdot f)=\tilde{O}(m)$ messages.
Recall, that given the labels of two fragments $L(F_i)$ and $L(F_j)$ it is now possible to 
compute the weight of the heaviest edge on the path connecting the fragments 
in $T_{\cF}$ by applying the decoder algorithm $\cD$. Once all vertices obtain the labels of their corresponding fragments, each vertex of $G$ can verify  (using the decoder~$\cD$) by communicating with its neighbors only,
that each inter fragment edge incident to it and not participating in 
$T_{\cF}$ forms a cycle when added to $T_{\cF}$ for which it is a cycle heavy 
edge. Following is a more detailed description of the algorithm.

\begin{enumerate}
\item
{\bf\underline{Phase 1}}

\begin{enumerate}
\item \label{send_ids} The source node $s$ (which initiates the algorithm) constructs a BFS tree rooted at $s$, computes the values $n$ and $d$, where $n$ is the number of nodes and $d$ is the depth of the BFS tree. (Note that $d$ is a 2-approximation to the diameter $D$ of the graph.) Subsequently, the source node $s$ broadcasts a signal over the BFS tree containing the values $n$ and $d$ and instructing each vertex to send its identifier to all its neighbors. 
This guarantees that all nodes know the values of $n$ and $d$ as well as the identifiers of their neighbors. 
In addition, to start the next step (i.e., Step 2.b) at the same time, the broadcast is augmented with a counter that is initialized by the source  $s$ to $d+1$ is decreased by one when delivered to a child in the BFS tree. Let $c(u)$ denote the counter received at node $u$. Before starting the next step each node $u$ waits for $c(u)$ rounds after receiving the counter.

\item \label{dom} 
Perform  Procedure $\DOM(k)$, with parameter $k=\max \{\sqrt n,d\}=\Theta(\sqrt n+D)$. 
The result is an MFC $\cF$, where each fragment 
$F\in \cF$ is of size $|V(F)|>k$ 
and depth $O(k)$, having a fragment leader and a distinct fragment identifier 
known to all vertices in the fragment.
\item \label{send_fragment} 
Each vertex sends its fragment identifier to all its neighbors.
\item  \label{identify_int_vs_ext_edges}
By comparing the fragment identifiers of its neighbors with its own fragment identifier, each vertex $v$ identifies the sets of  edges  $E_{F}(v)$ and $E_{F,T}(v)$.
\item \label{verify_internal_edges}
Verify that 
$\bigcup_{F\in \cF} E(F) = \bigcup_{F\in \cF} E_{F,T}$ 
by verifying at each vertex $v$ that $E_{F}(v)=E_{F,T}(v)$.
(Else return ``$T$ is not an MST''.)
\end{enumerate}

\item 
{\bf \underline{Phase 2}}
\begin{enumerate}

\item \label{count_fragments}  
Vertex $s$ counts the number of fragments, denoted by $f$. This is implemented by letting $s$ broadcast a signal over the BFS tree instructing the nodes to perform a convergecast 
over the BFS tree. During this convergecast the number of fragments is aggregated to the root $s$ by letting only fragment leader vertices increase 
the fragment counter. This guarantees that the fragment  counter at the root $s$ is $f$.
\item \label{count_inter_fragment}
Vertex $s$ counts the number of inter fragment 
edges induced by $T$ (i.e., the number of edges in $T_{\cF}$) by performing another convergecast over the BFS tree. Then, $s$ verifies that the number of edges is equal to $f -1$. 
(Else return ``$T$ is not an MST''.) 
\item \label{send_Tf}
Vertex $s$ broadcasts a signal over the BFS tree instructing all vertices to send  the description of all their incident edges in $T_{\cF}$ 
to $s$, by performing a convergecast over the BFS tree. (The edges of $T_{\cF}$ are 
all edges of $T$ that connect vertices from different fragments.)
\item \label{verify_tree}Vertex $s$ verifies that $T_{\cF}$ is a tree. 
(Else return ``$T$ is not an MST''.)
\end{enumerate}

\item 
{\bf \underline{Phase 3}}
\begin{enumerate} 
\item \label{route}
Each fragment leader sends a message to vertex $s$ over the BFS tree. 
Following these messages, all vertices save routing information regarding the fragment leaders.
I.e., if $v$ is a fragment leader and $v$ is a descendant of some other vertex $u$ in the BFS tree, then, 
after this step is applied, $u$ knows which of its children is on the path connecting it to  the fragment leader $v$.
\item \label{encode}
Vertex $s$ computes  the labels $L(F)$ for all vertices $F$ in $T_{\cF}$
(i.e., assigns a label to each of the fragments) 
using the encoder algorithm $\cE$. Subsequently, $s$ sends to each 
fragment leader its fragment label 
(the label of each fragment is sent to the fragment leader over the BFS tree 
using the routing information established in Step \ref{route} above).
(Recall that each label is encoded using $O(\log^2 n)$ bits.)
\item \label{send_labels_leaders}
Each fragment leader broadcasts its fragment label to all vertices in the fragment. The broadcast is done over the fragment edges.
\item \label{send_labels_neig}
Each vertex $v$ sends its fragment label to all its neighbors in other 
fragments. (Recall that $v$ already knows $E_F(v)$ by step 
\ref{identify_int_vs_ext_edges}.)
\item \label{verify_cycle_heavy}Each vertex $v$ verifies, for every 
neighbor $u$ such that $u$ does not belong to $v'$s fragment and 
$(u,v)\notin T$, that $\omega(v,u)\ge MAX(F(v), F(u))$. 
(The value $MAX(F(v), F(u))$ is computed by applying 
the decoder algorithm $\cD$ to the labels $L(F(v))$ and $L(F(u))$.) 
(Else return ``$T$ is not an MST''.)
\end{enumerate}
\end{enumerate}

\subsection{Correctness}

We now show that our MST verification algorithm correctly identifies whether 
the given subgraph $T$ is an MST. We begin with the following claim.

\begin{claim} \label{clm_fragmet_mst}
Let $T$ be a spanning tree of $G$ such that $T$ contains all edges of the MFC 
$\cF$ and $T_{\cF}$ forms an MST on the fragment graph of $G$ (with respect 
to  the MFC $\cF$). Then $T$ is an MST on $G$.
\end{claim}
\begin{proof} 
Since $\cF$ is an MST fragment collection, there exists an MST $T'$ of $G$ 
such that $T'$ contains all edges of $\cF$. Due to the minimality 
of $T'$, the fragment graph $T'_{\cF}$ induced by $T'$ necessarily forms 
an MST on $G_{\cF}$, the fragment graph of $G$. 
Hence we get that $\omega(T)=\omega(T')$, thus $T$ is an MST of~$G$.
\end{proof}

Due to the assumption that edge weights are distinct, we get:
\begin{observation} 
\label{unique_mst}
The MST of $G$ is unique.
\end{observation}
By combining Claim \ref{clm_fragmet_mst} and Observation \ref{unique_mst} 
we get the following.
\begin{claim} \label{show_mst}
A spanning tree $T$ of $G$ is an MST if and only if  $T$ contains all edges 
of the MFC $\cF$ and $T_{\cF}$ forms an MST on $G_{\cF}$.
\end{claim}

\begin{lemma}
\label{lem:correct}
The MST verification algorithm correctly identifies whether the given tree $T$ 
is an MST of the graph $G$.
\end{lemma}

\begin{proof} 
By Claim \ref{show_mst}, to prove the correctness of the algorithm it suffices 
to show that given an MST candidate $T$, the algorithm  verifies that: 
\begin{description}
\item{(1)} $T$ is a spanning tree of $G$,
\item{(2)} $T$ contains all edges of $\cF$, and
\item{(3)} $T_{\cF}$ forms an MST on  $G_{\cF}$.
\end{description}

Since $\cF$ as constructed by Procedure $\DOM$ 
in the first phase is an MFC, it spans all vertices of $G$. 
Step \ref{verify_internal_edges} verifies that 
$\bigcup_{F\in \cF} E(F) = \bigcup_{F\in \cF} E_{F,T}$, 
thus after this step, it is verified that $T$ does not contain a cycle that 
is fully contained in some fragment $F\in \cF$ (since every $F\in \cF$ 
is a tree). On the other hand, step \ref{verify_tree} verifies that $T$ 
does not contain a cycle that contains vertices from different fragments. 
Hence, the algorithm indeed verifies that $T$ is a spanning tree of $G$, 
and Property (1) follows. 
Property (2) is clearly verified by step \ref{verify_internal_edges} of the algorithm.
Finally, to verify that $T_{\cF}$ forms an MST on the fragment graph 
of $G$ it suffices to verify that inter-fragment edges not in $T_{\cF}$ 
are cycle heavy, which is done in step \ref{verify_cycle_heavy}.
Property (3) follows.
\end{proof}

\subsection{Complexity}

Consider first the Phase 1. Steps \ref{send_ids} and \ref{send_fragment} clearly take $O(D)$ time and $O(m)$ messages.
Step \ref{dom}, i.e., the execution of 
Procedure $\DOM$, requires $O(k\cdot \log ^*n)$ time and 
$O(E\cdot \log k +n\cdot \log^*n\cdot \log k)$ messages. 
(A full analysis appears in \cite{Vitaly}.) 
The remaining steps of the first phase are local computations performed by 
all vertices. Thus, since $k$ is set to $k=\Theta(\sqrt{n}+D)$,
the first phase of the MST verification algorithm requires 
requires $\tilde{O}(\sqrt{n}+D)$ time and 
$\tilde{O}(m)$ messages.

Observe now, that since the fragments are disjoint and each fragment contains at least $k$ vertices, 
we have the following.
\begin{observation} 
\label{num_of_fragments}
The number of MST fragments constructed during the first phase of the 
algorithm is $f \le n/k=O(\min\{\sqrt{n}, \frac{n}{D}\})$. 
\end{observation} 

The first two steps of phase 2, namely, steps \ref{count_fragments}
and \ref{count_inter_fragment} consist of simple broadcasts over the BFS tree, hence they require 
$O(D)$ time and $O(m)$ messages. Step \ref{send_Tf} consists of upcasting all 
the edges in $T_{\cF}$ 
to $s$. Note that the number of edges in $T_{\cF}$ can potentially be as high as $m$. However, given that the verification did not fail in Step 
\ref{count_inter_fragment}, we are guaranteed that the number of edges in $T_{\cF}$ is $f-1$.
Thus, Step \ref{send_Tf} amounts to upcasting $f-1$ messages and can therefore be performed in $O(D+f)$ time and $O(f D)$ messages.
Step \ref{route}  consists of an upcast of $f$ messages over the BFS tree and thus requires the same time and message complexities as that of step \ref{send_Tf}.
Step \ref{encode} consists of a BFS downcast of $f$ messages (each of size 
$\log^2n$); it therefore requires $O(D+f \log n)$ time and $O(f D\log n )$ messages. Step \ref{send_labels_leaders} consists of a broadcast of a label (of size $O(\log^2 n)$) in each of the MST fragments; this can be performed in  $O(k+\log n)$ time and
$O(n \log n)$ messages. Finally, step \ref{send_labels_neig} can be performed in $O(\log n)$ time and $O(m\log n)$ messages, and step 
\ref{verify_cycle_heavy} amounts to  local computations.
Table~1 below
summarizes the time and message complexities of the second and third phases 
of the algorithm.

\begin{table}[h!]
\label{t:table}
\noindent
\begin{tabular}{|l||*{3}{l|}}
\hline
Step
&\makebox[5em]{Description}&\makebox[5em]{Time}&\makebox[5em]{Messages}
\\\hline\hline
\ref{count_fragments},\ref{count_inter_fragment}& BFS convergecast & $O(D)$&$O(m)$ 

\\ \hline
\ref{send_Tf}, \ref{route} & BFS upcast of $f$ messages&$O(D+f)$&$O(f\cdot D)$ 

\\ \hline
\ref{verify_tree},\ref{verify_cycle_heavy} & Local computation&$0$ & none 

\\ \hline
\ref{encode} & BFS downcast of $f$ messages (each of size 
$\log^2n$)&$O(D+f\cdot \log n)$&$O(f\cdot \log n\cdot D)$ \\ \hline
\ref{send_labels_leaders} & Broadcast in each of the MST fragments & $O(k+\log n)$ & 
$O(n \cdot\log n)$

\\ \hline
 \ref{send_labels_neig} & 
Communication between neighbors & 
$O(\log n)$ & $O(\log n \cdot m)$ 
\\\hline
\end{tabular}
\caption{Time and message complexities of phases 2 and 3 of the algorithm. 
$D$ is the diameter of the graph, and $f$ is the number of MST\ fragments 
constructed in phase 1.}
\end{table}

\noindent
Combining 
the above arguments, and using the fact that $f =O(\min\{\sqrt{n}, \frac{n}{D}\})$ (see Observation \ref{num_of_fragments}), we obtain the following. 

\begin{lemma}
The  algorithm describes above requires $\tilde{O}(\sqrt{n}+D)$  time
and $\tilde{O}(m)$ messages.
\end{lemma}

Theorem \ref{thm:upper-bound} follows by combining the
 lemma above   with Lemma \ref{lem:correct}.

\section{Time Complexity Lower Bound}
\label{sec:time-lb}

In this section, we prove an $\tilde{\Omega}(\sqrt{n} + D)$ lower bound 
on the time required to solve the MST verification. 
Clearly, $\Omega(D)$ time is inevitable in the case of a single source node. 
The case in which all nodes start at the same time is also very simple.
Indeed, in this case, the $\Omega(D)$ time bound follows by considering 
a cycle $C$ of size $n$ with all edges having weight 1 except for two edges 
$e$ and $e'$ located at opposite sides. The given candidate spanning tree is 
$T=C\setminus \{e\}$. The two edges $e$ and $e'$ can have arbitrary weights, 
hence, to decide whether $T$ is an MST one needs to transfer information 
along at least half of the cycle.
This shows a lower bound of $\Omega(D)$ for the case where $D=\Theta(n)$. 
A similar argument can be applied for arbitrary values of $D$.

The difficult part is to prove a time lower bound of $\tilde{\Omega}(\sqrt{n})$.
We prove this lower bound for the case where all nodes start at the execution 
at the same time. The lower bound for the case of a single source node follows 
directly from this lower bound, relying on the leader election algorithm 
of \cite{KPPRT12}, which runs  in $\tilde{O}(D)$ time.

In the remaining of this section we consider the case where all nodes start the execution simultaneously at the same time, and prove a lower bound of $\tilde{\Omega}(\sqrt{n})$.
To show the lower bound, 
we first define a new problem named {\em vector equality}, and then show 
a lower bound for the time required to solve it. 
This is established in Section \ref{sec:VEprob}.
More specifically, for the purposes of this lower bound proof, we consider the collection of graphs denoted $F^2_m$, for
$m\ge 2$, as defined in \cite{PR_00}. Each graph graph $F^2_m$ consists of
$n = \Theta(m^4)$ vertices, and its diameter is $\Theta(m)$. In Section \ref{subsec:ITN}, we establish a time lower bound 
of $\tilde{\Omega}(m^2)=\tilde{\Omega}(\sqrt{n})$ for solving vector equality in each graph $F^2_m$.

In Section \ref{sec:lb_ver_prob}, for
$m\ge 2$, we consider a family $\cJ^2_m$ of weighted
graphs (these are weighted versions of the graph $F^2_m$), and show that any algorithm solving 
the MST verification problem on the graphs in $\cJ^2_m$ can also be used to 
solve the vector equality problem on $F^2_m$ with the same time complexity.
This establishes the desired $\tilde{\Omega}(\sqrt{n})$ time lower bound for the distributed MST verification problem.

\subsection{A lower bound for the vector equality problem}
\label{sec:VEprob}

\paragraph{The vector equality problem $\EQ(G,s,r,b)$.}
Consider a graph $G$ and two specified distinguished vertices $s$ and $r$, 
each storing $b$ boolean variables, 
$\bar{X}^s=\langle X_1^{s}, \dots, X_b^{s} \rangle $ and
$\bar{X}^r=\langle X_1^{r}, \dots,X_b^{r} \rangle $ respectively,
for some integer $b \ge 1$. An instance of the problem 
consists of initial assignments
$\INPUTS  = \{X_i^{s} \mid 1 \le i \le b\}$ and 
$\INPUTR = \{X_i^{r} \mid 1 \le i \le b\}$, 
where $X_i^{s},X_i^{r} \in \{0,1\}$,
to the variables of $s$ and $r$ respectively. Given such
an instance, the vector equality problem requires $r$ to decide whether 
$\bar{X}^s=\bar{X}^r$,  i.e., $X_i^s=X_i^r$ for every $1\le i\le b$.  

\subsubsection{The graphs $F^2_m$}
\label{ss:graphs}

Let us recall the collection of graphs denoted $F^2_m$, for
$m\ge 2$, as defined in \cite{PR_00}. (See Figure \ref{f:Ti}). The components used 
are the {\em ordinary path} $\cP$ on $m^2+1$ vertices, where
$V(\cP)= \{v_0,\dots,v_{m^2}\}$, 
$E(\cP) = \{(v_i, v_{i+1}) \mid 0 \le i \le m^2-1\}$,
and the {\em highway} $\cH$ on $m+1$ vertices, where
$V(\cH)= \{h_{im} \mid 0\le i\le m\}$ and
$E(\cH) = \{(h_{im},h_{(i+1)m}) \mid 0 \le i \le m-1\}$.
Each highway vertex $h_{im}$ is connected to the corresponding 
path vertex $v_{im}$ by a {\em spoke} edge $(h_{im},v_{im})$. 

The graph $F^2_m$ is constructed of $m^2$ copies of the path, 
$\cP^1,\dots,\cP^{m^2}$, and connecting them to the highway $\cH$.
The two distinguished vertices are $s=h_0$ and $r=h_{m^2}$.
The spoke edges are grouped into $m+1$ {\em stars} $S_i$, $0 \le i \le m$, 
with each $S_i$ consisting of the vertex $h_{im}$ and the $m^2$ vertices 
$v_{im}^1,\dots,v_{im}^{m^2}$ connected to it by spoke edges. Hence
$V(S_i) = \{h_{im}\} \cup \{v_{im}^1,\dots,v_{im}^{m^2}\}$
and $E(S_i) = \{(v^j_{im},h_{im}) \mid 1\le j \le m^2 \}$.
The graph $F^2_m$ consists of
$V(F^2_m) = V(\cH) \cup \bigcup_{j=1}^{m^2} V(\cP^j)$
and $E(F^2_m) = \bigcup_{i = 0}^{m} E(S_i) \cup
\bigcup_{j=1}^{m^2} E(\cP^j) \cup E(\cH)$.
Thus $F^2_m$ has $n = \Theta(m^4)$ vertices and diameter $\Theta(m)$.

\subsubsection{The lower bound for $EQ$}
\label{subsec:ITN}

Our goal now is to prove that solving the vector equality problem on $F_m^2$ with a $b=m^2$-bit input strings $\INPUT$ requires
$\Omega(m^2/B)$ time, assuming each message is encoded using $B$ bits. (Later, we shall take $B=O(\log n)$.)
\newline\indent
Consider some arbitrary algorithm $\cA_{EQ}$, and let $\varphi(\INPUT)$
denote the execution of $\cA_{EQ}$ on $m^2$-bit inputs $\INPUT$ on
the graph $F_m^2$. For simplicity, we assume that $m^2/2$ is an integer. For $1 \le i \le m^2/2$, 
define the {\em i-middle set}
of the graph $F_m^2$, denoted $M_i$, as follows. For every $1\le
j\le m^2$, define the {\em i-middle} of the path $\cP^j$ as
$M_i(\cP^j)=\{ v^j_l \mid i \le l \le m^2-i \}$.
Let $\beta(i)$ denote the least integer $\delta$ such that
$\delta m \ge i$ and $\gamma(i)$ denote the largest integer $\kappa$ such that 
$\kappa m \le m^2 -i$. Define the {\em i-middle} of $\cH$ as
$M_i(\cH)=\{ h_{jm} \mid \beta(i)\le j\le \gamma (i)\}$.
Now, the {\em i-middle set} of $F_m^2$ is the union of those middle sets,
$M_i = M_i(\cH) \cup \bigcup_j M_i(\cP^j).$
(See Fig.~\ref{f:Ti}.)
For $i=0$, the definition is slightly different, letting
$M_0=V\setminus \{h_0,h_{m^2}\}$.

\begin{figure} [htb]
\begin{center}

\centerline{\includegraphics[scale=0.60]{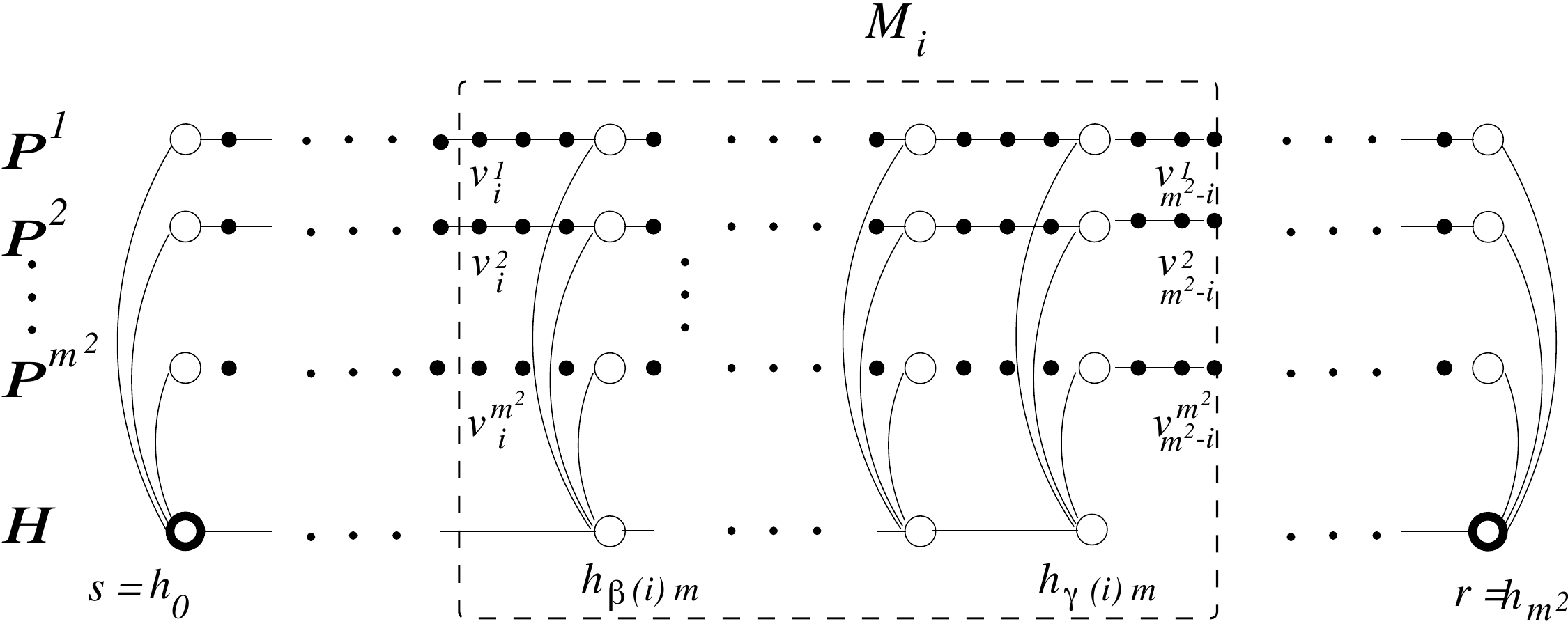}}

\caption[]{
\label{f:Ti}
\sf
The middle set $M_i$ in the graph $F^2_m$ }
\end{center}
\end{figure}

Denote the {\em state} of the vertex $v$ at the beginning of round $t$
during the execution $\varphi(\INPUT)$ on the input $\INPUT$ by
$\sigma(v,t,\INPUT)$. Without loss of generality, we may assume that in two different executions $\varphi(\INPUT)$ and
$\varphi(\INPUTPrime)$, a vertex reaches the same state at time
$t$, i.e., $\sigma(v,t,\INPUT) = \sigma(v,t,\INPUTPrime)$, 
if and only if  it receives the same sequence of messages on each of its incoming links; for different sequences, the states are distinguishable.
\newline\indent
For a given set of vertices $U = \{v_1, \dots, v_l\} \subseteq V$,
a {\em configuration} $C(U,t,\INPUT)$ is a vector
$\langle \sigma(v_1,t,\INPUT), \dots, \sigma(v_l, t, \INPUT) \rangle$
of the states of the vertices of $U$ at the beginning of
round $t$ of the execution $\varphi(\INPUT)$.
Denote by $\cC[U, t]$ the collection of all possible
configurations of the subset $U \subseteq V$ at time $t$ over all
executions $\varphi(\INPUT)$ of algorithm $\cA_{EQ}$ (i.e., on all legal
inputs $\INPUT$), and let $\rho[U,t] = |\cC[U,t]|$.
\newline\indent
Prior to the beginning of the execution (i.e., at the beginning of
round $t=0$), the input strings $\INPUT$ are known only to vertices 
$s$ and $r$ respectively. The rest of the vertices
are found in some initial state, described by the configuration
$C_{init} = C(M_0,0,\INPUT)$, which is independent of $\INPUT$.
Thus in particular $\rho[M_0,0] = 1$.
(Note, however, that $\rho[V,0] = 2^{2m^2}$.)

\begin{lemma}
\label{step}
For every $0 \le t < m^2/2$,
$\rho[M_{t+1},t+1] \le (2^B+1)^2 \cdot \rho[M_t,t]$.
\end{lemma}

\begin{proof}
The lemma is proved by showing that in round $t+1$ of the algorithm,
each configuration in $\cC[M_t,t]$ branches off into at most $(2^B+1)^2$
different configurations of $\cC[M_{t+1},t+1]$.

Fix a configuration $\hC \in \cC[M_t,t]$, and let $\delta=\beta(t+1)$ and 
$\kappa=\gamma(t+1)$.
The $(t+1)$-middle set $M_{t+1}$ is connected to the rest of the graph by the
highway edges $f_{\delta-1} = (h_{(\delta-1)m}, h_{\delta m})$ and 
$f_{\kappa} = (h_{\kappa m}, h_{(\kappa+1) m})$
and by the $2m^2$ path edges 
$e^j_t = (v_t^j,v_{t+1}^j)$,$e^j_{m^2-t} = (v_{m^2-t}^j,v_{m^2-t+1}^j)$
$1\le j\le m^2$. (See Fig. \ref{f:Tt+1}.)
\begin{figure} [htb]
\begin{center}

\centerline{\includegraphics[scale=0.6]{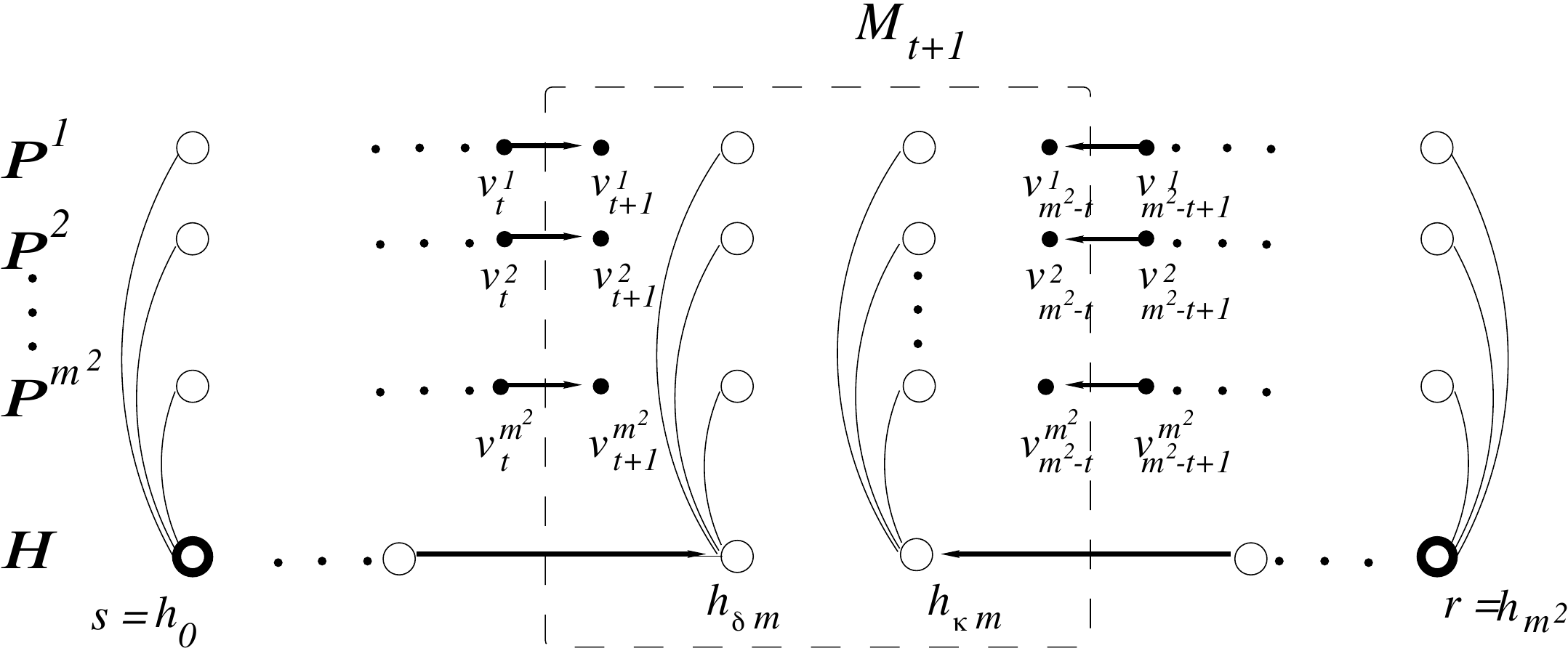}}

\caption[]{
\label{f:Tt+1}
\sf
The edges entering the middle set $M_{t+1}$. }
\end{center}
\end{figure}

Let us count the number of different configurations in
$\cC[M_{t+1},t+1]$
that may result of the configuration $\hC$.
Starting from the configuration $\hC$, each vertex $v_t^j$ is
restricted
to a single state, and hence it sends a single (well determined)
message over the edge $e^j_t$ to $v_{t+1}^j$
thus not introducing any divergence in the execution. Similarly, each vertex 
$v_{m^2-t+1}^j$  is restricted to a single state, and hence it sends a single
message over the edge $e^j_{m^2-t}$ to $v_{m^2-t}^j$.
The same applies to all the edges internal to $M_{t+1}$.
As for the highway edges $f_{\delta-1}, f_{\kappa}$, the vertices 
$h_{(\delta-1)m}$  and $h_{(\kappa+1)m}$ are not in the set $M_t$, 
hence they may be in any one of many possible states, 
and the values passed over these edges into the set $M_{t+1}$
are not determined by the configuration $\hC$. However, due to the
restriction of the $B$-bounded-message model, at most $(2^B+1)$
different behaviors of $f_{\delta-1}$ can be observed by $h_{\delta m}$ 
and at most $(2^B+1)$ different behaviors of $f_{\kappa}$ can be observed 
by $h_{\kappa m}$.
Thus altogether, the configuration $\hC$ branches off into at most
$(2^B+1)^2$ possible configurations 
$\hC_1,\dots,\hC_{2^B+1}^2 \in \cC[M_{t+1},t+1]$, 
differing only by the states $\sigma(h_{\delta m}, t+1, \INPUT)$, 
$\sigma(h_{\kappa m}, t+1, \INPUT)$.
The lemma follows.
\end{proof}

Applying Lemma \ref{step} and the fact that $\rho[M_0,0] = 1$,
we get the following result. 

\begin{corollary}
\label{cl:bound-rho}
For every $0 \le t \le m^2/2$, $\rho[M_t,t] \le (2^B+1)^{2t}$.
\end{corollary}

To prove the time lower bound for the vector equality problem we introduce 
the following restricted model of computation as defined in \cite{KUS_97}.
A {\em two party communication model} consists of parties $\partyA$ 
and $\partyB$ connected via a bidirectional communication link. 
Let $S_X,S_Y,S_Z$ be arbitrary finite sets. We say that Protocol $\Pi$ computes 
function $f:S_X\times S_Y \rightarrow S_Z$  if, when given input $a\in S_X$ 
known to party $\partyA$ and input $b\in S_Y$ known to party $B$, 
the parties are able to determine $f(a,b)$  by a sequence of bit exchanges. 
The following observation follows from \cite[1.14]{KUS_97}, and is used to prove Lemma \ref{l:EQ-LB}.

\begin{observation} \label{l:similar-executions}
Consider a protocol $\Pi $.  Suppose that there exist inputs 
$a,a'\in S_x$ for party $\partyA$ and $b,b'\in S_Y$ for party $\partyB$ 
for which   in the executions of $\Pi $ on input $(a,b)$ and 
on input $(a',b')$, identical sequences of bits are exchanged by both parties. 
Then  the same sequence of bits is exchanged in the execution of 
$\Pi $ on input $(a',b)$.
\end{observation}

\begin{lemma}
\label{l:EQ-LB}
For every $m\ge 1$, solving the vector equality problem 
$\EQ(F^2_{m},h_0,h_{m^2},m^2)$ requires $\Omega(m^2/B)$ time.
\end{lemma}

\begin{proof}
Assume, towards contradiction, that there exists a protocol $\Pi$ that 
correctly solves  $\EQ(F^2_{m},h_0,h_{m^2},m^2)$ and has worst case time complexity 
$T_{\Pi}< \frac{m^2}{2B}$. Let $M_{T_{\Pi}}$ be the middle set of $F_{m}^2$, 
as previously defined, corresponding to $i=T_\Pi$. Let $L_{T_{\Pi}}$ be 
the set of vertices  that reside on the left side of $M_{T_{\Pi}}$ in $F_m^2$ 
and $R_{T_{\Pi}}$ be the set of vertices that reside on the right side of 
$M_{T_{\Pi}}$. Note that $M_{T_{\Pi}}\cup L_{T_{\Pi}}\cup R_{T_{\Pi}}=V$ and 
$M_{T_{\Pi}}\cap L_{T_{\Pi}}=M_{T_{\Pi}}\cap  R_{T_{\Pi}} =
L_{T_{\Pi}} \cap  R_{T_{\Pi}}=\phi$.
Consider a simulation of protocol $\Pi$ by two parties $\partyA$ and $\partyB$.
The simulation works so that party $\partyA$ simulates the execution of $\Pi$ 
in $M_{T_{\Pi}}$ and in $L_{T_{\Pi}}$, and party $\partyB$ simulates the 
execution of $\Pi$ in $M_{T_{\Pi}}$ and in $R_{T_{\Pi}}$. 
At the end of the simulation, party $\partyB$ outputs the output of vertex $r$ 
in the execution of $\Pi$. Every step of the distributed protocol $\Pi$ 
is simulated by multiple bit exchanges between the parties $\partyA$ and $\partyB$. 
Party $\partyA$ sends to party  $\partyB$ all messages sent in $\Pi$ by vertices in 
$L_{T_{\Pi}}$ to vertices in $M_{T_{\Pi}}$ and $\partyB$ party sends to party $\partyA$ 
all messages sent in $\Pi$ by vertices in $R_{T_{\Pi}}$ to vertices in 
$M_{T_{\Pi}}$. Hence, the parties $\partyA$ and $\partyB$ have full knowledge of the 
configuration of vertices in $M_{T_{\Pi}}$ and are able to compute the 
configurations of vertices in $L_{T_{\Pi}}$ and vertices in $R_{T_{\Pi}}$ 
respectively. Thus, parties $\partyA$ and $\partyB$  correctly simulate $\Pi$.

Combining the assumption that $T_\Pi < \frac{m^2}{2B}$ with Corollary 
\ref{cl:bound-rho}, we get that the number of possible configuration of 
$M_{T_{\Pi}}$ in all possible executions of protocol $\Pi$ is smaller than 
$2^{m^2}$. 
Hence there exist inputs ${x},{x}'\in \{0,1\}^{m^2}$ such that ${x}\ne {x}'$ 
and protocol $\Pi$ reaches the same configuration of $M_{T_{\Pi}}$ when 
executed with input 
$\INPUTS=\INPUTR=x$ or $\INPUTS=\INPUTR=x'$. 
Denote by $\Pi_x$ the execution of $\Pi$ on $F_m^2$ with input 
$\INPUTS=\INPUTR=x$, and by $\Pi_{x'}$ the execution of $\Pi$ on $F_m^2$ with 
input $\INPUTS=\INPUTR=x'$.
During both simulations of $\Pi_x$ and $\Pi_{x'}$, parties $\partyA$ and 
$\partyB$ exchange the same sequence of messages (induced by the configuration 
of $M_{T_{\Pi}}$). By Obs. \ref{l:similar-executions} we get that in the 
simulation of $\Pi$ on input 
$\INPUTS=x', \INPUTR=x$, $\partyA$ and $\partyB$ exchange the same 
sequence of messages as in the simulation of $\Pi_x$.  Thus in both executions 
of protocol $\Pi$ ($\Pi_x$ and $\Pi$ with input $\INPUTS=x', \INPUTR=x$), 
vertex $r$ has identical initial configuration and receives the same 
sequence of messages, hence it reaches the same decision,
in contradiction to the correctness of $\Pi$.
\end{proof}

\subsection {A lower bound for the MST problem on $\cJ^2_m$}
\label{sec:lb_ver_prob}

In this section we use the results achieved in the previous subsection,
and show that 
the distributed MST verification
problem cannot be solved faster than $\Omega(m^2/B)$ rounds on weighted
versions of the graphs $F^2_m$.
In order to prove this lower bound, 
we recall the definition given in $\cite{PR_00}$ 
of a family of weighted graphs $\cJ^2_m$, based on $F_m^2$ but differing 
in their weight assignments.
Then, 
we show that any algorithm solving 
the MST verification problem on the graphs of $\cJ^2_m$ can also be used to 
solve the vector equality problem on $F^2_m$ with the same time complexity.
Subsequently, the lower bound for the distributed MST verification problem 
follows from 
the lower bound given in the previous subsection for the vector 
equality problem on $F^2_m$.

\subsubsection{The graph family $\cJ^2_m$}
The graphs $F^2_m$ defined earlier were unweighted.
In this subsection, we define for every graph $F^2_m$ a family of
weighted graphs
$\cJ^2_m = \{ J^2_{m,\gamma} = (F_m^2, \omega_\gamma) \mid
1 \le \gamma \le 2^{m^2}\}$,
where $\omega_\gamma$ is an edge weight function.
In all the weight functions $\omega_\gamma$, all the edges of the
highway $\cH$ and the paths $\cP^j$ are assigned the weight $0$.
The spokes of all stars except $S_0$ and $S_m$ are assigned
the weight 4. The spokes of the star $S_m$ are assigned the weight 2.
The only differences between different weight functions $\omega_\gamma$
occur on the $m^2$ spokes of the star $S_0$.
Specifically, each of these $m^2$ spokes is assigned a weight of
either 1 or 3; there are thus $2^{m^2}$ possible 
weight assignments. 
\newline\indent
Since discarding all spoke edges of weight 4 from the
graph $J^2_{m,\gamma}$ leaves it connected, and since every spoke
edge of weight 4 is the heaviest edge on some cycle in the graph, the
following is clear.

\begin{lemma}
\label{infWE}
No spoke edge of weight 4 belongs to the MST of $J^2_{m,\gamma}$,
for every $1\le\gamma \le 2^{m^2}$.
\end{lemma}

Let us pair the spoke edges of $S_0$ and $S_m$, denoting the $j$th
pair (for $1\le j\le m^2$) by
$PE^j = \{(s,v_0^j),(r,v_{m^2}^j)\}$.
The following is also clear.
\begin{lemma}
\label{finWE}
For every $1\le\gamma \le 2^{m^2}$ and $1 \le j \le m^2$, exactly
one of the two edges of $PE^j$
belongs to the MST of $J^2_{m,\gamma}$, namely, the lighter one.
Moreover, for every $m\ge 2$ and $1\le\gamma \le 2^{m^2}$,
all the edges of the highway $\cH$ and the paths $\cP^j$,
for $1\le j\le m^2$, belong to the MST of $J^2_{m,\gamma}$.
\end{lemma}

Note that the horizontal edges belong to the MST under any edge weight 
function. Of the remaining edges, exactly one of each pair will join the MST,
depending on the particular weight assignment to the spoke edges of
the star $S_0$.

\subsubsection{The lower bound}

\begin{lemma}
\label{reduc2}
Any distributed algorithm for  MST verification on the graphs of the
class $\cJ_m^2$, can be used to solve the
$\EQ(F_m^2,h_0,h_{m^2},m^2)$ problem on $F_m^2$ with the same
time complexity.
\end{lemma}

\begin{proof}
Consider an algorithm $\cA_{mst}$ for the MST verification problem, and 
suppose that we are given an instance of the $\EQ(F_m^2,h_0,h_{m^2},m^2)$
problem with input strings $\INPUT$ (stored in variables 
$\indVarS$ and $\indVarR$ respectively). We use the algorithm $\cA_{mst}$
to solve this instance of vector equality problem as follows.
Vertices $s=h_0$ and $r=h_m$ initiate the construction of an
instance of the MST verification by turning $F_m^2$ into a weighted graph from
$\cJ_m^2$, setting the edge weights and marking the edges participating in 
the MST candidate as follows:
for each $X_i^s \in \indVarS$, $1 \le i \le m^2$, vertex $s$ sets the weight
variable $W_i^{s}$ corresponding to the spoke edge $e_i \in
E(S_0)$,
to be
\begin{displaymath}
W_i^{s} =
\begin{cases}
3, &\text { if }  X_i^s = 1; \\ 
1, &\text { if }  X_i^s = 0.
\end{cases}
\end{displaymath}
In addition vertices $s$ and $r$ mark the edges participating in the MST 
candidate (that we wish to verify) in the following manner:  
for each $X_i^s \in \indVarS$, $1 \le i \le m^2$, vertex $s$ sets 
the indicator variable $Y_i^{s}$ corresponding to the spoke edge 
$e_i \in E(S_0)$,
to be
\begin{displaymath}
Y^s_i =
\begin{cases}
0, &\text { if }  X_i^s = 1; \\
1, &\text { if }  X_i^s = 0.
\end{cases}
\end{displaymath}
for each $X_i^r \in \indVarR$, $1 \le i \le m^2$, vertex $r$ sets 
the indicator variable $Y_i^{r}$ corresponding to the spoke edge 
$e_i \in E(S_m)$,
to be
\begin{displaymath}
Y^r_i =
\begin{cases}
1, &\text { if }  X_i^r = 1; \\
0, &\text { if }  X_i^r = 0.
\end{cases}
\end{displaymath}
The rest of the graph edges are assigned fixed weights as specified above.
All path edges and highway edges are marked 
as participating in the MST candidate. All spoke edges not belonging to stars 
$S_0, S_m$ are marked as not participating in the MST candidate.  
Note that the weights of all the edges except those connecting $s$ to its 
immediate neighbors in $S_0$ do not depend on the particular input instance 
at hand. Similarly, for all edges except the spoke edges belonging to 
$S_0,S_m$ the indicator variable for participating in the MST candidate 
does not depend on the instance at hand.  Hence a constant number of rounds
of communication between $s,r$ and their  $S_0,S_m$ neighbors suffices 
for performing this assignment; $s$ assigns its edges weights and indicator 
variables according to its input string $\INPUTS$, and needs to send at most 
a constant number of messages to each of its neighbors in $S_0$, to notify it 
about the weight and the indicator variable of the spoke edge connecting them.
(Same argument holds for vertex $r$).
Every vertex $v$ in the network,
upon receiving the first message of algorithm $\cA_{mst}$, assigns the
values defined by the edge weight function $\omega_\gamma$ to its
variables $W_i^v$ and the corresponding indicator variable $Y_i^v$ 
as described above.
(As discussed earlier, this does not require $v$ to know $\gamma$, as
its assignment is identical under all weight functions $\omega_\gamma$,
$1 \le \gamma \le 2^{m^2}$).
From this point on, $v$ may proceed with executing algorithm
$\cA_{mst}$ for the MST verification problem.

Once algorithm $\cA_{mst}$ terminates,  vertex $r$ 
determines its output for the vector equality problem by stating that both input vectors are equal if and only if the MST verification algorithm verified that the given MST candidate is indeed a minimum spanning tree.
By Lemma \ref{finWE}, the lighter edge of each pair $PE^j$, for
$1\le j\le m^2$, belongs to the MST; 
thus, by the construction of the MST candidate and the weight assignment 
to the edges of $F^2_m$ the MST candidate forms a minimum spanning tree 
if and only if  the input vectors $\INPUT$ for the vector equality problem are equal.
Hence the resulting algorithm  correctly solves the given instance
of the vector equality  problem.
\end{proof}

Combined with Lemma \ref{l:EQ-LB}, we now have:
\begin{theorem}
For every $m \ge 1$, any distributed algorithm for solving  MST verification 
problem on the graphs of the family $\cJ_m^2$ 
requires $\Omega(m^2/B)$ time.
\end{theorem}

\begin{corollary}
\label{cl:sqrtn2}
Any distributed algorithm for the MST verification problem 
requires $\Omega(\sqrt{n}/B)$ time on
some $n$-vertex graphs of diameter $O(n^{1/4})$.
\end{corollary} 

Theorem \ref{thm:lower-bound-time} follows.

\section{Message Complexity Lower Bound}
\label{sec:msg-lb}

We prove a message complexity lower bound of $\Omega(m)$ on the 
{\em Spanning Tree (ST) verification} problem, which is a relaxed version 
of the MST verification problem defined as follows. 
Given a connected graph $G=(V,E,\omega)$ and a subgraph $T$ 
(referred to as the {\em ST candidate}), we wish to decide whether $T$ 
is a spanning tree\footnote{Equivalently, we may consider also disconnected 
graphs, and require $T$ to be a spanning forest of $G$.}
of $G$ (not necessarily of minimal weight). 
Clearly, a lower bound on the {\em ST verification} problem also applies 
to the {\em MST verification} problem. Since spanning tree verification 
is independent of the edge weights, for convenience 
we consider unweighted networks throughout this lower bound proof.

The case of single source node is easy. Here, given a $G$ with spanning tree 
$T$, if the execution does not send a message on some edge 
$e\in G\setminus\{T\}$, then we simply break $e$ to two edges by adding 
another vertex in the middle of $e$. This brings us to the case where $T$ 
is no longer spanning, but the execution (and hence the output of nodes) 
remains the same.

In what follows, we consider the somewhat more difficult case where 
all nodes start the execution at the same time. 
We begin with a few definitions. Recall that a \textit{protocol} is a local 
program executed by all vertices in the network. 
In every step, each processor performs local computations, 
sends and receives messages, and changes its state according to the 
instructions of the protocol. A protocol achieving a given task should work 
on every network $G$  and every ST candidate $T$.

Following \cite{AGVP_90}, we
denote the {\em execution} of protocol $\Pi$ on network $G(V,E)$ with 
ST candidate $T$ by  $EX(\Pi, G, T)$.
The {\em message history} of an execution $EX=EX(\Pi, G, T)$
is a sequence
describing the messages sent during the execution $EX$.
Consider a protocol $\Pi$, two graphs $G_0(V,E_0)$ and $G_1(V,E_1)$ over the 
same set of vertices $V$, and two ST candidates $T_0$ and $T_1$ for $G_0$ 
and $G_1$ respectively, 
and the corresponding executions $EX_0=EX(\Pi, G_0, T_0)$ and 
$EX_1=EX(\Pi, G_0, T_1)$. We say that the executions are {\em similar} 
if their message history is identical.

\begin{figure*}[htb]
\begin{center}
\begin{minipage}{\textwidth}
\centerline{\includegraphics[scale=0.6]{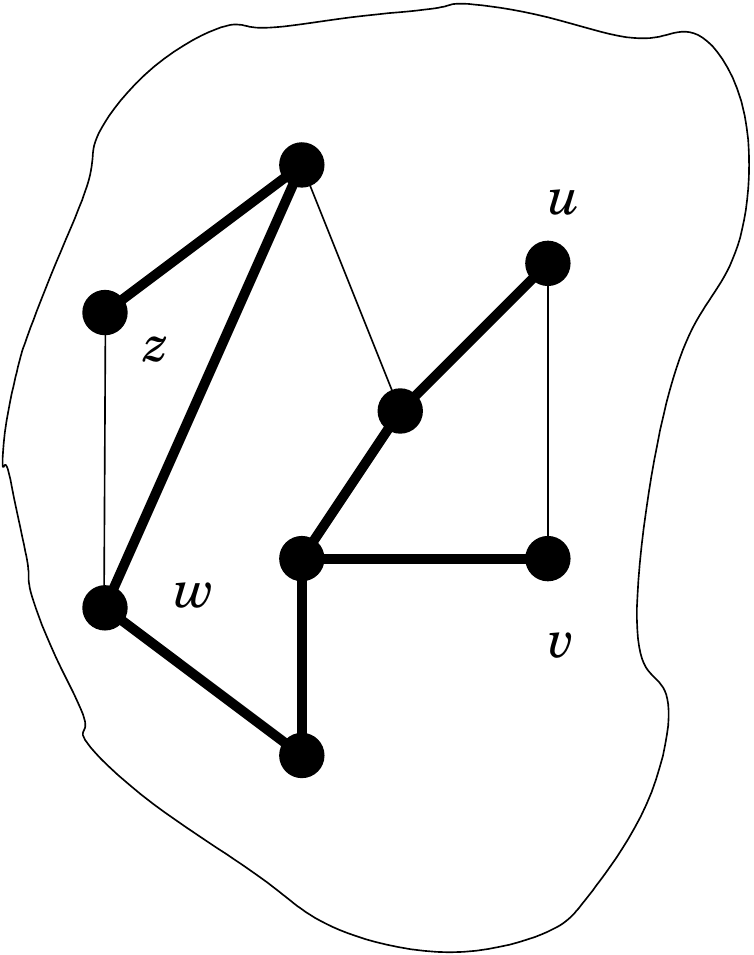}}
\end{minipage}
\caption[]{
\label{f:Graph_G_MST}
\sf Graph $G$ with ST candidate $T$ (the bold edges belong to $T$)}
\end{center}
\end{figure*}

Let $G=(V, E)$ be a graph (together with an assignment $\ID$
of vertex identifiers), 
$T$ be a subgraph and let $e_1=(u,v)$ and $e_2=(z,w)$ be two of its edges.
(See example in Figure \ref{f:Graph_G_MST}.) 
Let  $G'=(V',E')$ be some copy of $G=(V,E)$, where the identifiers of 
the vertices in $V'$ are not only pairwise distinct but also distinct 
from the given identifiers on $V$. Consider the following graphs 
$G^2$ and $G^X$ and the subgraph~$T^2$. 
\begin{itemize}
\item
 Graph $G^2$ is simply $G^2=(V^2, E^2)=G\cup G'=(V\cup V', E\cup E')$.
The subgraph $T^2$ of $G^2$ is defined as the union of the two copies of $T$, 
one in $G$ and the other in $G'$. See example in Figure~\ref{f:Graph_G2_MST}.
Note that $G^2$ is not connected, hence it is not a legal input to our problem.
\item
The graph $G^X$ is a ``cross-wired'' version of $G^2$. 
Formally,
$G^X=(V^2, E^X)$, where 
$E^X=E^{2}\sminus \{(u,v),(z',w')\} \cup \{(u,w'), (v,z')\}$. 
(Observe that for $e_1,e_2\notin T$, $T^2$ is also a subgraph of $G^X$.) 
See example in Figure~\ref{f:Graph_G2_e_MST}.
\end{itemize}

\begin{figure*}[htb]
\begin{center}
\begin{minipage}{\textwidth}
\centerline{\includegraphics[scale=0.6]{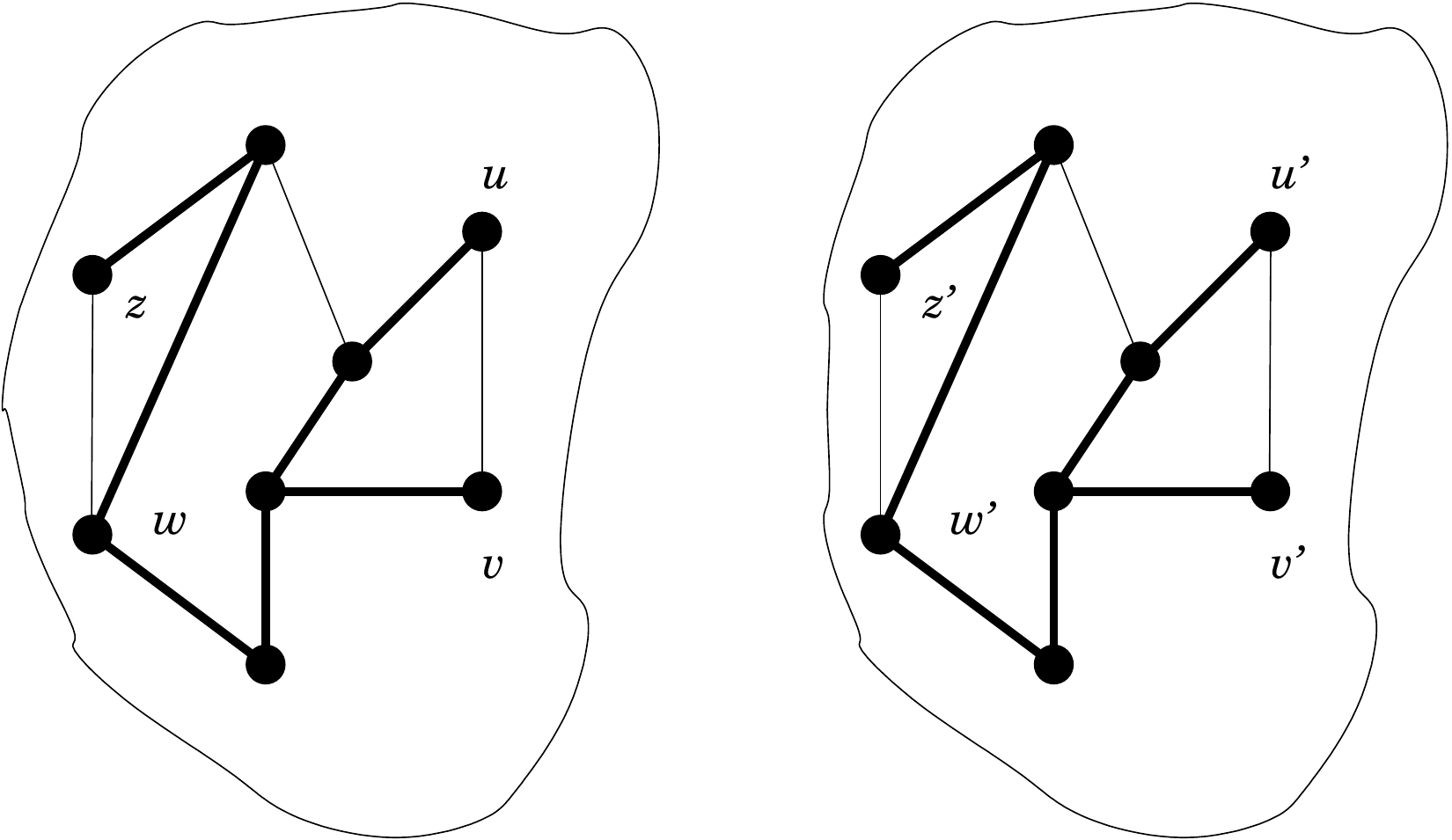}}
\end{minipage}
\caption[]{
\label{f:Graph_G2_MST}
\sf Graph $G^2$ with ST candidate $T^2$ (the bold edges belong to $T^2$)}
\end{center}
\end{figure*}

\begin{figure*}[htb]
\begin{center}
\begin{minipage}{\textwidth}
\centerline{\includegraphics[scale=0.6]{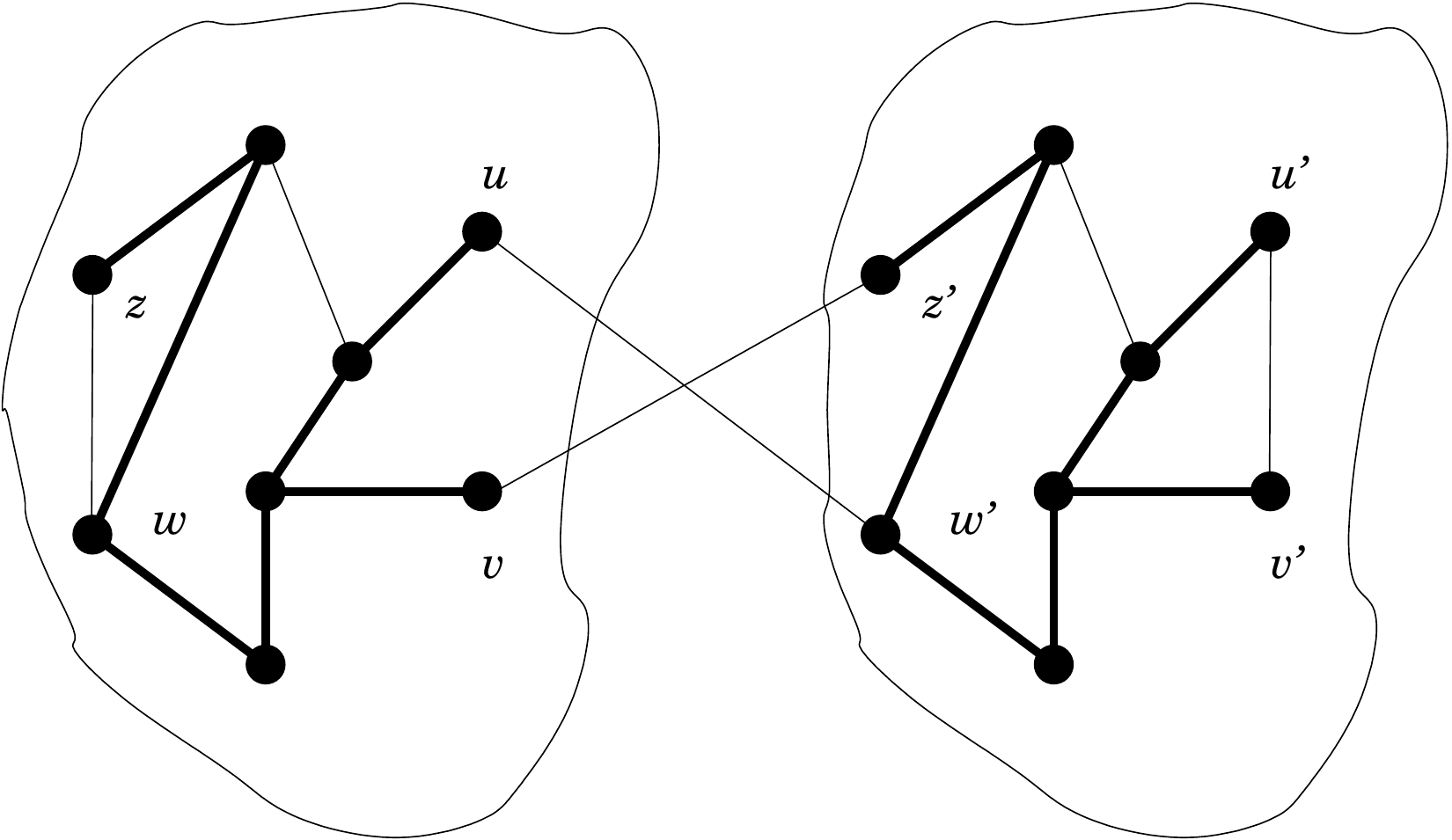}}
\end{minipage}
\caption[]{
\label{f:Graph_G2_e_MST}
\sf Graph $G^X$ with ST candidate $T^2$ (the bold edges belong to $T^2$)}
\end{center}
\end{figure*}

Let $\Pi$ be a protocol that correctly solves the ST verification problem.
Fix $G$ to be some arbitrary graph, fix a copy $G'$ of $G$, 
fix a spanning tree $T$ of $G$, 
and consider  the execution of $\Pi$ on either of the graphs $G,G', G^2$ and $G^X$ 
with the ST candidates $T,T,T^2$ and $T^2$, respectively. 
We stress  that $G^2$ (with candidate $T^2$) is not a valid input for the ST 
(or the MST) verification problem since it is not connected. 
Still, we can consider the execution $EX(\Pi, G^2,T^2)$, 
without requiring anything from its output.

\begin{lemma} 
\label{lem_similarity} 
Let $e_1\in E\setminus E(T)$ and $e'_2\in E'\setminus E'(T)$, such that no 
message is sent over the edges $e_1$ and $e'_2$ in execution $EX(\Pi,G^2,T^2)$. 
Then executions $EX(\Pi, G^2,T^2)$ and $EX(\Pi, G^X,T^2)$ 
are similar.
\end{lemma}

\begin{proof}
We show that in both executions each vertex sends and receives identical 
sequences of messages in each communication round of the protocol. 
Note that at each round the messages sent by some vertex $x$ are dependent 
on $x$'s topological view (neighbors of $x$), $x$'s initial input 
(its identity and the indicator variables of the edges incident to $x$), and 
the set of messages sent and received by $x$ in previous communication rounds.
Denote by $EX^2$ and  $EX^X$ executions $EX(\Pi, G^2, T^2)$ and  
$EX(\Pi, G^X, T^2)$ respectively. Note that any vertex 
$x\in V^2\sminus\{u,v,z',w'\}$ has identical topological view and 
identical initial input in both executions. Vertex $u$ has identical 
initial input and  identical number of neighbors in both  executions. 
Although the communication link connecting $u$ to $v$ in $G^2$ connects 
$u$ to $w'$ in $G^X$, vertex $u$ is initially unaware of this difference 
between the executions since it does not know the identifiers of its neighbors.
(The same holds for vertices $v$, $z'$ and $v'$.) 
The proof is by induction on $r$, 
the number of communication rounds of protocol $\Pi$.

\noindent\textbf{Induction base:} For $r=0$. In the first communication round, 
the messages sent by each vertex depend solely on its topological view 
and initial input. Let us analyze the sequence of messages sent by vertices 
in $V$ (the vertices of graphs $G^2$ and $G^X$ that belong to the first copy 
of $G$). Following are the possible cases.

\begin{itemize}
\item Vertex $x\notin \{u,v\}$: Vertex $x$ has identical topological view 
and identical initial input in both execution, thus it sends identical 
sequences of messages in the first round of both executions.
\item
Vertex $u$: As mentioned above, although in execution $EX^X$ vertex $u$ 
is connected to $w'$ instead of $v$, it has no knowledge of this difference. 
Thus $u$  sends identical sequences of messages over each of its communication 
links. The fact that no messages are sent over edge $e$ in execution $EX^2$, 
implies that in execution $EX^X$ no message is sent by $u$ to its 
neighbor $w'$. Thus, $u$ sends identical sequences of messages in the 
first communication round of both executions.
\item
Vertex $v$: can be analyzed in the same manner as vertex $u$.
\end{itemize}
\noindent
The above shows that  vertices in $V$ send the same sequence of messages 
in the first communication round of both executions. The induction base claim 
follows by applying the same argument on vertices of $V'$.

\noindent\textbf{Inductive step:} Can be shown using a similar case analysis as 
in the induction base.
\end{proof}

Theorem \ref{thm:lower-bound-msgs} follows as a consequence of the following 
Lemma.
\begin{lemma}
\label{clm_msg_lb}
Execution 
$EX(\Pi, G^2,T^2)$ requires $\Omega (|E^2\smallsetminus T^2|)$ messages.
\end{lemma}

\begin{proof}
Assume, towards contradiction, that there exists a protocol $\Pi$ 
that correctly solves the ST verification problem for every graph $G$ and 
ST candidate $T$, such that execution $EX(\Pi, G, T)$ sends fewer than
$|E\sminus T|/2$ messages over edges from $E \sminus T$.

For the rest of the proof we fix $G=(V,E)$ to be an arbitrary connected graph 
and denote the ST candidate by $T$.
We take $T$ to be a spanning tree and not just any subgraph.
(See Figure \ref{f:Graph_G_MST}).
 
Consider the graph $G^2$ as previously defined with ST candidate 
$T^2=\{e=(x,y)\in T\}\cup \{e'=(x',y') \mid e=(x,y)\in T\}$ 
(See Figure \ref{f:Graph_G2_MST}).  

Then by the assumption on $\Pi$, 
execution $EX^2=EX(\Pi, G^2, T^2)$ sends fewer than 
$|E^2\sminus T^2|/2$ messages over edges from $E^2\sminus T^2$. 
Hence there exist $e_1=(u,v)$ and $e'_2=(w',z')$ such that 
$e_1,e'_2\in E^2\sminus T^2$ and no message is sent over $e_1$ and $e'_2$ 
in execution $EX^2$. Consider the graph $G^X$ with ST candidate $T^2$ 
as previously defined (See Figure \ref{f:Graph_G2_e_MST}).

By Lemma \ref{lem_similarity}, executions $EX^2$ and $EX^X=EX(\Pi, G^X, T^2)$ 
are similar. Note that $e_1,e'_2 \notin T^2$, 
thus $T^2$ is not a spanning tree of $G^X$ (since the two copies of $G$   
contained in $G^X$ are connected solely by edges $e_1$ and $e'_2$). 
Since $\Pi$  correctly solves the ST verification problem, the output 
of all vertices in $EX^X$ is ``0'' (i.e., the given ST candidate $T^2$ is not 
a spanning tree of the graph $G^X$). 

On the other hand, consider the execution 
$EX=(\Pi,G, T)$ with ST candidate $T$. Note that $EX$ is exactly the 
restriction of $EX^2$ on the first copy of $G$ contained in $G^2$. 
Since $G^2$ contains two disconnected copies of $G$ the output of all vertices 
in execution $EX^2$ will be identical to the output of the same vertices 
in $EX$ (since in both executions the vertices have identical topological view 
and the input variables contain identical values). 
Since executions $EX^2$ and $EX^X$ are similar, the output of $EX$ is ``0'', 
in contradiction to the correctness of $\Pi$.
\end{proof}

\end{document}